\newif\ifllncs

\ifllncs
    \documentclass{llncs}
\else
    \documentclass{article}
    \usepackage{fullpage}
    \usepackage{palatino} 
    \usepackage{mathpazo} 
\fi

\usepackage[T1]{fontenc}
\usepackage{amsmath,amssymb,mathtools}
\usepackage{amsthm}
\usepackage{needspace}
\usepackage{booktabs,tabularx}
\usepackage{xcolor}
\usepackage[most]{tcolorbox}
\usepackage[colorlinks=true]{hyperref}
\usepackage[nameinlink,noabbrev]{cleveref}
\usepackage{dsfont}
\hypersetup{
  pdftitle={Quantum Time-Lock Puzzles in the Quantum Random Oracle Model},
  pdfsubject={Quantum time-lock puzzles and oracle-depth security},
  bookmarksnumbered=true,
  linkcolor=blue, citecolor=magenta, urlcolor=blue}
\usepackage[
  backend=biber,
  backref=true,
  style=alphabetic,
  maxalphanames=5,
  maxbibnames=99,
  doi=false,
  url=false
]{biblatex}
\makeatletter
\let\qtlp@cite\@cite
\makeatother

\allowdisplaybreaks

\newtcolorbox{protocol}[2][]{enhanced,breakable,lines before break=8,
  colback=white,colframe=black,colbacktitle=white,coltitle=black,
  title={#2},
  boxrule=.4pt,arc=0pt,left=5pt,right=5pt,top=4pt,bottom=4pt,
  toptitle=4pt,bottomtitle=3pt,titlerule=.4pt,
  before skip=\medskipamount,after skip=\medskipamount,
  title after break={#2\enspace\textnormal{(continued)}},
  before upper={\setlength{\parindent}{0pt}},#1}

\newcommand{\ket}[1]{\lvert #1\rangle}
\newcommand{\bra}[1]{\langle #1\rvert}
\newcommand{\braket}[2]{\langle #1\mid #2\rangle}
\newcommand{\proj}[1]{\ket{#1}\!\bra{#1}}
\newcommand{\Tr}{\operatorname{Tr}}
\newcommand{\E}{\mathbb{E}}
\newcommand{\N}{\mathbb{N}}

\newcommand{\Had}{\mathrm H}
\newcommand{\PGen}{\mathsf{PGen}}
\newcommand{\PSolve}{\mathsf{PSolve}}

\newcommand{\puz}{Z}
\newcommand{\sol}{s}
\newcommand{\msk}{\mathsf{msk}}

\newcommand{\adversary}{\mathcal{A}}
\newcommand{\Real}{\mathsf{Real}}

\newcommand{\Hyb}{\mathsf{Hyb}}

\newcommand{\reg}[1]{\textcolor{gray}{\mathbf{#1}}}
\newcommand{\secp}{\lambda}
\newcommand{\negl}{\mathrm{negl}}

\ifllncs
    \spnewtheorem{construction}[theorem]{Construction}{\bfseries}{\rmfamily}

    \spnewtheorem{claim}[theorem]{Claim}{\bfseries}{\itshape}
\else
    \newtheorem{theorem}{Theorem}[section]
    
    \newtheorem{definition}[theorem]{Definition}
    \newtheorem{lemma}[theorem]{Lemma}
    \newtheorem{claim}[theorem]{Claim}
    \crefname{claim}{claim}{claims}
    \Crefname{claim}{Claim}{Claims}

    \newtheorem{corollary}[theorem]{Corollary}

    \crefname{fact}{fact}{facts}
    \Crefname{fact}{Fact}{Facts}
    \newtheorem{construction}[theorem]{Construction}
    \crefname{construction}{construction}{constructions}
    \Crefname{construction}{Construction}{Constructions}

\fi

\title{Quantum Time-Lock Puzzles in the Quantum Random Oracle Model}
\author{%
  \begin{tabular}{@{}c@{\hspace{3em}}c@{}}
    {\large Prabhanjan Ananth} &
    {\large Yao-Ting Lin}\thanks{This work was done while the author was a PhD student at UCSB.} \\[3pt]
    {\small\texttt{prabhanjan@cs.ucsb.edu}} &
    {\small\texttt{yl6463@columbia.edu}} \\[2pt]
    UCSB & Columbia University
  \end{tabular}%
}
\date{}

\begin{document}
\maketitle
\thispagestyle{plain}

\begin{abstract}
A time-lock puzzle allows a sender to hide a message in a puzzle such that recovering the message requires substantially more sequential computation than the time required to generate the puzzle, even when parallel computation is allowed. Applications of time-lock puzzles include timed-release encryption, sealed-bid auctions, electronic voting, fair contract signing, coin flipping, and Byzantine consensus. However, time-lock puzzles are known to be impossible in the classical random oracle model. 

To overcome the classical barrier, in this work we consider quantum time-lock puzzles, in which the puzzle itself is a quantum state. Our construction in the quantum random oracle model achieves generation in one oracle round, solving in at most $T$ oracle rounds, and security against polynomial-width quantum adversaries of depth $o(T)$ for every polynomially bounded delay $T=T(\secp)$, resolving an open problem posed by Mahmoody, Moran, and Vadhan (2011).
\end{abstract}

\section{Introduction}
\label{sec:introduction}

The time-lock puzzles (TLPs) we study allow a sender to hide a message in time polynomial in the security parameter $\secp$ and $\log T$, while recovering it requires a prescribed amount $T$ (the \emph{delay parameter}) of \emph{sequential} computation, even with polynomially many parallel processors. Introduced by Rivest, Shamir, and Wagner~\cite{RSW96}, TLPs enable timed-release encryption for sealed-bid auctions and delayed access to confidential records. Time-lock techniques also underlie timed commitments for fair contract signing and three-round zero-knowledge protocols in the timing model~\cite{BN00}. Homomorphic TLPs support electronic voting~\cite{MT19}, while non-malleable TLPs yield fair non-interactive multiparty coin flipping and auctions~\cite{FKPS21}. Batchable TLPs have also found applications to Byzantine consensus~\cite{SLMNPT23}. Known positive results, including variants with preprocessing or sublinear generation time, fall into the following categories:
\begin{itemize}
    \item \textbf{Algebraic assumptions.} Constructions rely on sequential squaring in RSA groups~\cite{RSW96}; sequential squaring together with hard subgroup membership in imaginary-quadratic class groups~\cite{TCLM21}; or the bilinear isogeny shortcut assumption~\cite{BDF21}. The latter two approaches require preprocessing or setup, and the isogeny construction also uses random oracles.
    \item \textbf{Post-quantum assumptions.} Lattice-based constructions use circular small-secret LWE together with an iteratively sequential function~\cite{AMZ24}, or circular LWE together with bounded-space worst-case non-parallelizing languages~\cite{BG25}.\footnote{Ignoring security-parameter factors,~\cite{AMZ24} achieves $\sqrt{T}$ generation time in the plain model, or polylogarithmic online generation after preprocessing;~\cite{BG25} achieves $T^\varepsilon$ generation for any constant $\varepsilon>0$. Earlier preprocessing variants from subexponential LWE appear in~\cite{BGJPVW16}, with hiding time measured from publication of the preprocessing parameters.} Decomposed LWE and Ring-LWE yield fully succinct randomized encodings~\cite{AMR25}, which can be combined with the generic compiler below to obtain another construction. These routes are plausibly post-quantum secure when both the lattice and sequentiality assumptions hold against quantum algorithms.
    \item \textbf{Generic assumptions.} Bitansky et al.~\cite{BGJPVW16} construct TLPs from succinct randomized encodings and the existence of worst-case non-parallelizing languages. Instantiating the encodings using indistinguishability obfuscation gives puzzles with generation time polylogarithmic in $T$.
\end{itemize}
We do not know how to construct TLPs from symmetric-key assumptions alone. Mahmoody, Moran, and Vadhan~\cite{MMV11} show that, in the classical random oracle model, any time-lock puzzle generated using $q$ oracle queries can be solved in $O(q)$ sequential rounds of random-oracle queries (with polynomially many parallel queries allowed in each round). This rules out the required generation--solving gap\footnote{That is, the sequential time required to solve a puzzle should be significantly greater than the time required to generate it.} and yields black-box separations from one-way permutations and collision-resistant hash functions.

Quantum resources have enabled a line of works to overcome classical impossibility results. A prominent example is \emph{unclonable cryptography}, which uses quantum states to realize cryptographic tasks that are impossible to achieve classically~\cite{Wie83,Aar09}. For time-lock puzzles, Mahmoody, Moran, and Vadhan~\cite[Section~4]{MMV11} left open whether their impossibility result extends to quantum honest parties: 
\begin{quote}
   {\em ... On the other hand, when the honest parties are quantum, the lower bound question is still open as well... }
\end{quote}
This open question has not been resolved in the past 15 years. 
\par In this work, we study \emph{quantum time-lock puzzles} (QTLPs), in which the puzzle itself can be a quantum state. A QTLP consists of two quantum algorithms:
\begin{itemize}
    \item \textbf{Generator ($\PGen$).} Takes a security parameter $1^\secp$, a delay parameter $T$, and a solution bit $\sol$, and produces a quantum puzzle $\ket{\puz}$.
    \item \textbf{Solver ($\PSolve$).} Takes the puzzle and recovers $\sol$.
\end{itemize}
We require the following properties:
\begin{itemize}
    \item \textbf{Correctness.} $\PSolve$ recovers the bit except with negligible probability.
    \item \textbf{Efficiency.} Generation takes time polynomial in $\secp$ and $\log T$, while solving takes time polynomial in $\secp$ and $T$.
    \item \textbf{Sequential security.} Any adversary running in sequential time $o(T)$ only has negligible advantage in distinguishing puzzles for the two solution bits, even with polynomially (in $T$) many parallel processors.
\end{itemize}
In the quantum random oracle model (QROM), all the parties may query a random function in superposition, and the delay is measured in terms of rounds of parallel oracle queries: honest solving uses at most $T$ rounds, while security is required against adversaries using $r = o(T)$ rounds. This leads us to ask:
\begin{quote}
\emph{Can we construct quantum time-lock puzzles in the quantum random oracle model?}
\end{quote}

\subsection{Our Results}
Our construction uses a product of single-qubit BB84 states and random oracles to obtain an efficient generator and a sequential quantum solver. We state the parameters of a simple instantiation of \Cref{con:lock} below.

\begin{theorem}[Informal]
\label{thm:main-informal}
There is a quantum time-lock puzzle scheme in the QROM that satisfies the following guarantees for any polynomial delay $T = T(\secp)$:
\begin{itemize}
    \item \textbf{Generation.} The puzzle can be generated in time $\mathrm{poly}(\secp,\log T)$ using $\secp+1$ classical oracle queries in one parallel round. In particular, the puzzle contains $O(\secp^2\log T)$ single-qubit BB84 states and $\secp^2+1$ classical bits.
    \item \textbf{Solving.} The puzzle can be solved in time $\mathrm{poly}(\secp,T)$ using at most $T$ oracle rounds, with at most $\secp$ quantum queries per round and negligible correctness error.
    \item \textbf{Security.} For every polynomial $p=p(\secp)$ and every $r=o(T)$, any computationally unbounded adversary given one copy of the puzzle and making at most $r$ rounds of $p$-parallel quantum oracle queries guesses the uniformly random solution bit with probability at most $\tfrac12+\mathrm{negl}(\secp)$.
\end{itemize}
\end{theorem}

Afshar, Chung, Hsieh, Lin, and Mahmoody~\cite{ACHLM23} show that, in the QROM, any classical puzzle generated using $q$ classical random-oracle queries can be solved in $O(q)$ sequential rounds of random-oracle queries (with polynomially many parallel queries allowed in each round), even when the honest solver is allowed to make quantum queries. They also rule out perfectly correct schemes with a quantum generator and a classical solver. Their results leave open the setting in which the puzzle itself is quantum and the honest solver is quantum.

Our construction crucially uses quantum puzzles to bypass the impossibility results for classical puzzles. In particular, the attacks in~\cite{MMV11,ACHLM23} rely on running the honest solver multiple times on the same classical puzzle. These attacks therefore do not directly apply to our setting, in which the adversary receives only a single copy of the quantum puzzle. Moreover, our generator requires only classical oracle queries.

\paragraph{AI disclosure.}
We used OpenAI's ChatGPT (GPT-5.6 Sol Pro) and Codex (with the ultra reasoning setting) to assist with developing the construction, revising proofs, and preparing the manuscript. Our starting point was a human-proposed construction whose puzzle-generation time was linear in $T$ and a candidate construction in the Haar cipher model\footnote{By the Haar cipher model we mean oracle access to $U = \sum_{a} \proj{a} \otimes U_a$, where $U_a$ are independently sampled Haar-random unitaries.} that captures the main ideas of our actual construction, though we did not have a proof for it. After we supplied these constructions to ChatGPT, it proposed an approach with generation time polylogarithmic in $T$ in the QROM. The initial proposal from Codex used mutually unbiased bases (MUBs)~\cite{WF89}, whose instantiation was involved and required more sophisticated tools. Seeking an elementary proof, we iterated with ChatGPT to obtain the conceptually simpler construction using only BB84 states presented here.

An earlier version of \Cref{lem:hidden-pair-list} had a structurally different proof. After identifying its connection to monogamy-of-entanglement games, we revised the proof to use an operator inequality from~\cite{TFKW13} as a black box.

We also used these tools to draft and revise text throughout the manuscript. The authors substantially edited every section and take responsibility for the final manuscript, including its mathematical claims and proofs.

\subsection{Technical Overview}

\paragraph{Classical approach.}
We begin with a natural classical approach to develop intuition for our quantum construction. Although the known impossibility results rule out secure time-lock puzzles in the classical random oracle model, understanding why this approach fails suggests how quantum resources might help.

For a desired integer delay $T$, a natural classical generator samples $a\gets[T]$ uniformly and publishes $G(a)$, where $G:[T]\to\{0,1\}^{\secp}$ is a random oracle and $\secp$ is the security parameter. We consider delays $T=T(\secp)$ bounded by a polynomial in $\secp$. Generation requires one query, while a sequential search takes up to $T$ queries. However, an adversary can {\em copy} the puzzle, partition the candidates among many processors, and test them in parallel. Since $T$ is polynomially bounded, a polynomial-width adversary can even test all candidates in one round. We therefore seek a puzzle whose use in one search cannot simply be replicated across arbitrarily many independent searches.

\paragraph{First attempt: using an unclonable state.}
Quantum states suggest a way to prevent this copying attack. In particular, if the puzzle is an unclonable quantum state, then the above copying attack would fail. The unclonable quantum state will be defined with respect to a family of unitaries defined next. Let $\{U_a\}_{a\in[N]}$ be a public family of efficiently computable unitaries acting on $n$ qubits. We postpone specifying the unitaries and the token length $n$. We set $N=\Theta(T^2)$. Jumping ahead, the honest solver runs a variant of Grover search with search space $[N]$, which requires $O(\sqrt{N}) = O(T)$ rounds of queries.

The generator samples independent, uniformly random $a\gets[N]$ and $x\gets\{0,1\}^n$, and prepares exactly one copy of the token
\[
  \ket\psi=U_a\ket x.
\]
We initially set aside the solution bit and consider a weaker recovery task: the honest solver should recover the hidden pair $(a,x)$, while an adversary using $o(T)$ rounds of polynomially many parallel queries should have recovery probability bounded away from $1$. We will later amplify this guarantee and convert it into indistinguishability security.

\paragraph{Correctness.}
How should the honest solver proceed? A first idea is to try the candidate bases one at a time: for each $a'\in[N]$, apply $U_{a'}^\dagger$ to the token and measure in the computational basis. When $a'=a$, this returns $x$. But the solver does not know $x$, so it cannot recognize a successful attempt. To provide a verification procedure, the generator additionally publishes a tag
\[
  z := H(a,x),\qquad H:\{0,1\}^* \to \{0,1\}^{\secp},
\]
where $H$ is a random oracle. A similar idea appears in the copy-protection scheme for point functions by Coladangelo, Majenz, and Poremba~\cite[Construction~3]{CMP20}. To test a candidate $a'$, the solver applies $U^\dagger_{a'}$ to $\ket{\psi}$ and coherently checks whether the resulting value $x'$ satisfies $H(a',x')=z$. When $a'\neq a$, the check rejects with probability $1-2^{-\secp}$, averaged over the random oracle, since $H(a',\cdot)$ is independent of the tag $z$. By the gentle measurement lemma, this check causes only a negligible disturbance. The solver then applies $U_{a'}$ to undo the basis change and approximately restore the original puzzle state. Since there are only $N=\mathrm{poly}(\secp)$ candidates, the solver can reuse the same quantum puzzle to test successive candidates until one passes, with only negligible error.

A careful reader might notice that one can use Grover search or amplitude amplification to speed up the above honest solver algorithm. However, since the search space $N$ is polynomial, a naive attempt might cause an inverse-polynomial error. In more detail, let us consider Grover search. Suppose the angle between the initial state and the unmarked subspace (orthogonal to the target state) in Grover search is $\phi$. Then, after $g$ Grover iterations, the resulting state has amplitude $\sin((2g+1)\phi)$ on the target state. If $\tilde{g} := \frac{\pi}{4\phi} - \frac{1}{2}$ is an integer, then after $\tilde{g}$ iterations, the state will land entirely on the target state. Otherwise, after $\lceil\tilde{g}\rceil$ iterations, the success probability $\sin^2((2\lceil\tilde{g}\rceil+1)\phi)$ could be $1 - 1/\mathrm{poly}(\secp)$. Therefore, our construction instead uses an “exact” variant of Grover search introduced in~\cite{BHMT02}. We explain the idea in what follows. We prepare a slightly ``worse'' initial state whose angle with the unmarked subspace, $\theta := \frac{\pi/2}{2\lceil\tilde{g}\rceil+1}$, is slightly smaller than $\phi = \frac{\pi/2}{2\tilde{g}+1}$. By design, the new angle $\theta$ satisfies $(2\lceil\tilde{g}\rceil+1)\theta = \pi/2$. Therefore, after $\lceil\tilde{g}\rceil \leq \tilde{g}+1 = O(T)$ Grover iterations, we can recover the target exactly. A concrete initialization is the following.
Initialize the search register $\reg A$ in
\[
  \ket{\mathsf{Init}}=\sqrt{1-N\sin^2\theta}\,\ket\bot
  +\sin\theta\sum_{a'\in[N]}\ket{a'},
\]
where $\bot$ is a dummy label. Registers $\reg B$ and $\reg R$ hold $\ket\psi=U_a\ket x$ and $\ket{0^{\secp}}$, respectively. To explain the first verification operation, write
\[
  \begin{aligned}
  U_{a'}^\dagger\ket\psi := \sum_y \alpha_{a',y} \ket y ,
  \qquad \ket{\mathsf{Dummy}} :=\sqrt{1-N\sin^2\theta}\,
  \ket\bot_{\reg A}\ket\psi_{\reg B}\ket{0^{\secp}}_{\reg R}.
  \end{aligned}
\]
In particular, $U_a^\dagger\ket{\psi}=\ket{x}$, so $\alpha_{a,y}=1$ if $y=x$ and $\alpha_{a,y}=0$ otherwise. Below, sums range over $a'\in[N]$ and $y\in\{0,1\}^n$; the dummy branch is unchanged.
\begin{itemize}
  \item Start with 
  \[
  \ket{\mathsf{Init}}_{\reg A} \ket{\psi}_{\reg B} \ket{0^{\secp}}_{\reg R}.
  \]
  \item Apply $U_{a'}^\dagger$ to $\reg B$, controlled by $\reg A$:
  \[
    \ket{\mathsf{Dummy}}+\sin\theta\sum_{a',y}\alpha_{a',y}
    \ket{a'}_{\reg A}\ket y_{\reg B}\ket{0^{\secp}}_{\reg R}.
  \]
  \item Compute the hash $H(a',y)$ in $\reg R$:
  \[
    \ket{\mathsf{Dummy}}+\sin\theta\sum_{a',y}\alpha_{a',y}
    \ket{a'}_{\reg A}\ket y_{\reg B}\ket{H(a',y)}_{\reg R}.
  \]
  \item Flip the phase on non-dummy branches when $H(a',y)=z$:
  \[
    \ket{\mathsf{Dummy}}+\sin\theta\sum_{a',y}(-1)^{[H(a',y)=z]}\alpha_{a',y}
    \ket{a'}_{\reg A}\ket y_{\reg B}\ket{H(a',y)}_{\reg R}.
  \]
  \item Uncompute the hash, resetting $\reg R$:
  \[
    \ket{\mathsf{Dummy}}+\sin\theta\sum_{a',y}(-1)^{[H(a',y)=z]}\alpha_{a',y}
    \ket{a'}_{\reg A}\ket y_{\reg B}\ket{0^{\secp}}_{\reg R}.
  \]
  \item Undo the basis change by applying controlled $U_{a'}$ to $\reg B$:
  \begin{align*}
    & \ket{\mathsf{Dummy}}+\sin\theta\sum_{a',y}(-1)^{[H(a',y)=z]}\alpha_{a',y}
    \ket{a'}_{\reg A}\bigl(U_{a'}\ket y\bigr)_{\reg B}
    \ket{0^{\secp}}_{\reg R} \\
    = & \ket{\mathsf{Dummy}} + \sin\theta\sum_{a'} \ket{a'}_{\reg A} \otimes \underbrace{\bigl(U_{a'} \sum_y (-1)^{[H(a',y)=z]} \alpha_{a',y} \ket y_{\reg B} \bigr)}_{\ket{\Psi_{H,a'}}}
    \ket{0^{\secp}}_{\reg R}
  \end{align*}
\end{itemize}
Since $H(a,x)=z$, the correct branch satisfies
\[
  \ket{\Psi_{H,a}}=-U_a\ket{x}=-\ket{\psi},
\]
so the token is restored with phase $-1$. For an incorrect branch $a'\neq a$, we show that $\ket{\Psi_{H,a'}}$ is negligibly close to $\ket{\psi}$ on average over $H$. Indeed,
\[
  \ket{\Psi_{H,a'}}-\ket{\psi}
    =-2U_{a'}\sum_{y:H(a',y)=z}\alpha_{a',y}\ket{y},
\]
where we used $U_{a'}\sum_y\alpha_{a',y}\ket{y} =U_{a'}U_{a'}^\dagger\ket{\psi}=\ket{\psi}$. Since the oracle row $H(a',\cdot)$ is independent of the token and the tag $z$, each $y$ satisfies $H(a',y)=z$ with probability $2^{-\secp}$. Therefore,
\[
  \mathbb{E}_H\!\left[
    \bigl\|\ket{\Psi_{H,a'}}-\ket{\psi}\bigr\|_2^2
  \right]
  =4\sum_y|\alpha_{a',y}|^2
    \Pr_H[H(a',y)=z]
  =4\cdot 2^{-\secp}.
\]
Thus, on the input state above, this verification operation approximates the ideal Grover phase oracle, which negates the correct branch and leaves all other branches unchanged. Averaging over all branches, its expected squared error is
\[
  4(N-1)\sin^2\theta\cdot 2^{-\secp}
  \leq 4\cdot 2^{-\secp} 
  = \mathrm{negl}(\secp).
\]
The solver alternates verification with reflection about $\ket{\mathsf{Init}}$ for $g$ iterations, using $2g$ rounds. It then measures $\reg A$, aborts if the outcome is $\bot$, and otherwise measures $\reg B$. We bound the accumulated error and show that the solver recovers the hidden pair $(a,x)$ except with negligible probability. See~\Cref{lem:correctness} for details.

\paragraph{Security.} We next address security. We require the set of $n$-qubit unitaries $\{U_a\}_{a \in [N]}$ to satisfy the following ``incompatible'' condition
\[
  \bigl|\bra x U_a^\dagger U_{a'}\ket{x'}\bigr|
  \leq \varepsilon
  \qquad \text{whenever }(a,x) \neq (a',x') \in [N] \times \{0,1\}^n,
\]
where $n \approx \secp\log(N)$ and $\varepsilon = \mathrm{negl}(\secp)$. We will explain how to construct such sets efficiently in a later paragraph. In what follows, we gradually strengthen the security in different settings.

\paragraph{(i)~The list-one-wayness lemma in the plain model.}  We first consider the following experiment in the plain model. The list-one-wayness lemma (\Cref{lem:hidden-pair-list}) states that given the description of $\{U_a\}_{a \in [N]}$ and one copy of $U_a \ket x$ with $(a,x)$ chosen uniformly at random, any computationally unbounded adversary outputting a list $L$ of $p$ candidate pairs satisfies
\begin{equation} \label{eq:list_recovery_bound}
  \Pr[(a,x) \in L] 
  \leq \delta_{p,\varepsilon,N} := \frac{1+(p-1)\varepsilon}{N}.
\end{equation}
Crucially, the dependency on $p$ is suppressed by $\varepsilon$. \Cref{eq:list_recovery_bound} is the key for establishing sequentiality, and we crucially exploit the incompatibility of the bases and the property of quantum states. Consider the following classical analogue in which the unitaries are replaced by permutations $\{\pi_a\}_{a\in [N]}$. Note that $\varepsilon = 1$ in this case. Given $\pi_a(x)$, the solver can guess a permutation in each branch, say $a_i$, evaluate $x_i = \pi^{-1}_{a_i}(\pi_a(x))$, and output $(a_i,x_i)$. Whenever $a_i=a$, we also have $x_i=x$. Thus, the resulting list contains $(a,x)$ with probability $p/N$, unlike~\Cref{eq:list_recovery_bound}.

To show~\Cref{eq:list_recovery_bound}, we first express the density matrix received by the adversary in terms of sums of projectors. Define: 
\[
  P_{a,x}=U_a\proj x U_a^\dagger,
  \qquad A_L=\sum_{(a,x)\in L}P_{a,x}.
\]
Ultimately, upper bounding the operator norm of $A_L$ will help us prove an upper bound on $\Pr[(a,x) \in L]$. Proving upper bounds on sums of operators shows up in monogamy of entanglement games, which in turn have been influential for unclonable cryptography. In particular, we rely upon the operator inequality of Tomamichel, Fehr, Kaniewski, and Wehner~\cite[Lemma~2]{TFKW13}, which, when specialized to projectors and cyclic permutations, gives
\[
  \begin{gathered}
  \left\|\sum_{j=1}^{p} P_j\right\|
  \leq\sum_{t=0}^{p-1} \max_{j\in[p]}
  \left\|P_j P_{\pi_t(j)}\right\|,\\
  \pi_t(j)=1+((j-1+t) \bmod p).
  \end{gathered}
\]
The identity permutation contributes $1$. Every other permutation pairs distinct projectors and therefore contributes at most $\varepsilon$. Consequently,
\[
  \|A_L\| \leq 1 + (p-1)\varepsilon.
\]
Now describe the adversary's arbitrary measurement by a POVM $\{M_L\}$, whose outcome specifies its output list. Averaging over the uniformly random hidden pair gives
\[
  \Pr[(a,x)\in L]
  =\frac1{N2^n}\sum_L\Tr(M_LA_L)
  \leq\frac{1 + (p-1)\varepsilon}{N},
\]
using $\sum_LM_L=I$.

\paragraph{(ii)~Adding the verification tag in the QROM.} 
Next, we consider proving the following statement in the QROM. Given a quantum token $U_a \ket{x}$ and a tag $z = H(a,x)$ with $(a,x)$ sampled uniformly at random, any adversary requires at least $\Omega(T)$ rounds of queries to retrieve $(a,x)$ with a constant probability. With~\Cref{eq:list_recovery_bound} in hand, such a result can be proven by a standard round-by-round reprogramming argument~\cite{BBBV97,Unr13,AHU19}. In particular, we show that each round of queries can increase the success probability of the adversary by at most $O(\sqrt{\delta_{p,\varepsilon,N}})$. As a result, achieving a constant success probability requires $\Omega(\sqrt{N}) = \Omega(T)$ rounds of queries when $p = \mathrm{poly}(\secp)$ and $\varepsilon = \mathrm{negl}(\secp)$.

In more detail, let $r$ denote the number of sequential oracle-query rounds and $p$ the maximum number of queries made in parallel in each round, where $p$ is polynomial in $\secp$. We will compare this experiment with one in which $z$ is \emph{independent} of $a,x$ and the oracle $H$. Since the two oracles need differ only at $(a,x)$, the key is to bound the probability that measuring one round's query register yields $(a,x)$.

First suppose that $H$ and $z$ are sampled independently of each other and of $(a,x)$. Similar to the idea of the one-way-to-hiding lemma~\cite{Unr13,AHU19}, measuring the $p$-parallel query registers produces a list of at most $p$ pairs. The adversary could produce this list from the token alone by simulating the independent oracle and tag. Thus, \Cref{eq:list_recovery_bound} implies
\[
  \Pr[(a,x) \text{ appears in the list}]
  \leq \delta_{p,\varepsilon,N}.
\]

To apply this bound to the real experiment, start with independently sampled $H$ and $z$, and let $\widetilde H$ agree with $H$ everywhere except that $\widetilde{H}(a,x)=z$. In other words, $\widetilde{H}$ is obtained by reprogramming $H$ to $z$ at $(a,x)$. Giving the adversary access to $\widetilde H$ produces the real experiment. We compare hybrids in which the first $j$ query rounds use $H$ and the remaining rounds use $\widetilde H$, for $j=0,\ldots,r$. The first hybrid is the real experiment, while in the last hybrid the oracle and tag are independent of the hidden pair. Adjacent hybrids differ on only one query round.

Immediately before that differing round, the adversary has interacted only with the independent oracle $H$. Measuring its $p$ query registers would therefore produce a list governed by the bound above. Hence the probability that at least one measured register equals $(a,x)$ is at most $\delta_{p,\varepsilon,N}$. By the quantum hybrid argument~\cite{BBBV97,AHU19}, the corresponding oracle-reprogramming bound is at most $2 \sqrt{\delta_{p,\varepsilon,N}}$, and summing over $r$ hybrid hops gives $2r\sqrt{\delta_{p,\varepsilon,N}}$.

In the last hybrid, both $H$ and $z$ are independent of the hidden pair. The adversary's oracle interaction can therefore be simulated using randomness independent of that pair. Applying \Cref{lem:hidden-pair-list} again, now with a list of size one, bounds its probability of recovering $(a,x)$ by $1/N$. Combining the two bounds gives
\[
  \Pr[\text{adversary outputs } (a,x)]
  \leq \beta_{r,p,\varepsilon,N}
  := \frac1N+2r\sqrt{\delta_{p,\varepsilon,N}}.
\]
For any polynomial query width $p$, our parameter choices give
\[
  \beta_{r,p,\varepsilon,N}\leq\frac1N+O(r/T)+\negl(\secp).
\]
Thus, when $r=o(T)$, the success probability is $o(1)$. See~\Cref{lem:hidden-pair-list-oracle} for details.

\paragraph{(iii)~Parallel repetition.}
The single-instance bound $\beta_{r,p,\varepsilon,N}$ is not yet negligible as $N = \mathrm{poly}(\secp)$. Indeed, an adversary can guess a basis uniformly and measure the token in that basis, recovering the pair with probability at least $1/N$.

We amplify security by taking $\ell=\secp$ independent instances. For each $i\in[\ell]$, independently sample $a_i,x_i$, prepare one token $U_{a_i}\ket{x_i}$, and publish $z_i=H_i(a_i,x_i)$, where the $H_i$ are independent random oracles. Each instance forms a fixed interactive quantum experiment supplying a token, oracle access for $r$ rounds, and a final pair-recovery test. Since we are working in an information-theoretic setting, Gutoski's parallel-repetition theorem~\cite[Theorem~4.9]{Gut10} bounds the probability of passing all independent tests by the product of their individual optimal success probabilities, even under arbitrary joint quantum operations. Consequently,
\[
  \Pr[\text{recover every }(a_i,x_i)]\leq\beta_{r,p,\varepsilon,N}^{\secp}.
\]
This uses the interactive repetition theorem; it does not follow merely from bounding each marginal success probability. Since $N\geq4$, for $r=o(T)$ we have $\beta_{r,p,\varepsilon,N}\leq1/4+o(1)$, and in particular $\beta_{r,p,\varepsilon,N}\leq1/2$ for sufficiently large $\secp$. Recovering the entire tuple therefore has negligible probability. See~\Cref{cor:hidden-pair-parallel-repetition} for details.

\paragraph{(iv) From one-wayness to indistinguishability.}
So far we have shown that an $r$-round adversary, given $U_{a_1}\ket{x_1},\ldots,U_{a_{\ell}}\ket{x_{\ell}},H_1(a_1,x_1),\ldots,H_\ell(a_{\ell},x_{\ell})$, cannot predict $(a_1,x_1,\ldots,a_{\ell},x_\ell)$ except with negligible probability. To convert such one-wayness to indistinguishability, we rely on the one-way-to-hiding lemma again. We now propose our QTLP construction that satisfies indistinguishability security for the solution bit. Let $F:\{0,1\}^*\to\{0,1\}$ be a random oracle independent of all the tag oracles $H_i$. To generate a puzzle for $\sol \in \{0,1\}$:
\begin{itemize}
  \item Sample independent pairs $(a_i,x_i)$ for $i\in[\ell]$.
  \item Prepare exactly one copy of each token $U_{a_i}\ket{x_i}$.
  \item Compute all tags $z_i = H_i(a_i,x_i)$ and the masked bit
  \[
    c=s\oplus F(a_1,x_1,\ldots,a_\ell,x_\ell).
  \]
  \item Output the tokens, the tags, and $c$.
\end{itemize}
The puzzle contains $\ell n = \secp n = O(\secp^2\log T)$ qubits and $\secp^2+1$ classical bits. Generation takes time $\mathrm{poly}(\secp,\log T)$, while honest solving takes time $\mathrm{poly}(\secp,T)$ and at most $T$ oracle rounds, including the final mask query.

An adversary that distinguishes the oracle-derived mask from an independent random bit must place sufficient query weight on the hidden tuple. Measuring an appropriate query would then recover that tuple, contradicting the amplified one-wayness bound. The parallel-query formulation follows the same oracle-hybrid argument; see also~\cite{AHU19}. Quantitatively, for a uniformly random solution bit,
\[
  \Pr[\text{adversary outputs } s]
  \leq\frac12+2r\sqrt{p\,\beta_{r,p,\varepsilon,N}^{\secp}}
  =\frac12+\negl(\secp)
\]
for polynomial $p$ and $r=o(T)$. See~\Cref{lem:security} for details.

\paragraph{Instantiating the set of unitaries.} All that remains is to identify the unitary family satisfying $\bigl|\bra x U_a^\dagger U_{a'}\ket{x'}\bigr|
  \leq\varepsilon$. These overlap conditions have an elementary BB84 implementation. Identify $a\in[N]$ with a $k$-bit string and encode it by repeating that string $\secp$ times:
\[
  C(a)=\underbrace{a\Vert\cdots\Vert a}_{\secp\text{ copies}}.
\]
This determines the token length $n=\secp k=O(\secp\log T)$. Define
\[
  U_a=\bigotimes_{j=1}^{n}\Had^{\,C(a)_j},
\]
where $\Had$ is the Hadamard gate. Distinct codewords differ in at least $\secp$ positions, so their cross-basis overlaps are at most $\varepsilon=2^{-\secp/2}$. States with different $x$-labels in the same basis are orthogonal. Thus the required condition holds for every pair of distinct hidden pairs, using only products of single-qubit BB84 states. Here $U_a$ abbreviates the codeword-indexed unitary $U_{C(a)}$ in~\Cref{def:code-unitaries}; code-selected BB84 bases also appear in~\cite[Section~4.1]{FS17}.

\paragraph{Related work.} 
Mahmoody, Moran, and Vadhan~\cite{MMV11} prove an impossibility result for TLPs in the random oracle model. Afshar et al.~\cite{ACHLM23} extend this study to the QROM and prove impossibility when either the generator or the solver is classical. Hidden-basis encodings appear in work of Damg\aa rd, Pedersen, and Salvail on quantum key uncertainty~\cite{DPS04}. The idea of combining binary linear codes with BB84 states appeared in the work of Fehr and Salvail for quantum authentication and key recycling~\cite[Section~4.1]{FS17}.

\section{Preliminaries}
\label{sec:prelim}

\paragraph{Notations.}
Throughout the work, we write $\secp$ to denote the security parameter. We define $[m]$ to be $\{1,\ldots,m\}$. We assume that the reader is familiar with the basics of quantum information.

\subsection{Quantum Random Oracle Model}

In the quantum random oracle model (QROM), by \emph{query round} we mean a round of parallel oracle queries. Specifically, let $H$  be a random oracle and $p \in \N$. The query operator is defined as $O_H \ket{x} \ket{y} := \ket{x} \ket{y \oplus H(x)}$. One round of $p$-parallel queries to $H$ refers to applying the operator $O_H^{\otimes p}$. We write $r$ for the number of query rounds and $p$ for the query \emph{width}, that is, the number of parallel queries in each round. When multiple oracles are available, superposition queries across different oracles are not allowed.

\subsection{Quantum Time-Lock Puzzles in the QROM}

We use the syntax of time-lock puzzles in~\cite{RSW96}, with the relaxation that puzzles are allowed to be quantum states. Throughout, delay parameters $T=T(\secp)$ are polynomially bounded integers satisfying $T=\omega(1)$.

\begin{definition}[Quantum Time-Lock Puzzles]
A quantum time-lock puzzle (QTLP) in the quantum random oracle model is a pair of quantum oracle algorithms $(\PGen^{(\cdot)}, \PSolve^{(\cdot)})$.
    \begin{itemize}
        \item $\PGen^H(1^\secp,T(\secp),\sol)\to\ket{\puz}$: The generation algorithm takes as input a security parameter $1^\secp$, a delay parameter $T(\cdot)$, and a solution bit $\sol\in\{0,1\}$, and outputs a quantum puzzle $\ket{\puz}$.
        \item $\PSolve^H(1^\secp,T(\secp),\ket{\puz})\to\sol'$: The solver takes as input a security parameter $1^\secp$, a delay parameter $T(\cdot)$, and a quantum puzzle $\ket{\puz}$, and outputs $\sol'\in\{0,1,\bot\}$, where $\bot$ denotes an abort.
    \end{itemize}
\end{definition}

\begin{definition}[Correctness]
\label{def:qtlp-correctness}
A QTLP is \emph{correct} if for all polynomials $T(\cdot)$ and $\sol\in\{0,1\}$, there exists a negligible function $\negl$ such that for all $\secp$,
    \[
    \Pr\left[
    \sol' = \sol : \begin{aligned}
        &\ket{\puz} \gets \PGen^{H}(1^\secp,T(\secp),\sol), \\
        &\sol' \gets \PSolve^{H}(1^\secp,T(\secp),\ket{\puz})
        \end{aligned}
    \right] \geq 1-\negl(\secp),
    \]
where the probability is over the random oracle $H$ and all internal randomness and measurements.
\end{definition}

\begin{definition}[Security]
\label{def:qtlp-security}
A QTLP is \emph{secure} if, for all polynomials $T(\cdot)$, $p(\cdot)$, any function $r=o(T)$, and any (possibly computationally unbounded) adversary $\adversary$ making at most $r$ rounds of $p$-parallel queries, there exists a negligible function $\negl$ such that for all $\secp$,
    \[
    \Pr_H\left[\sol' = \sol: \substack{
      \sol \gets \{0,1\}, \\
      \ket{\puz} \gets \PGen^{H}(1^\secp,T(\secp),\sol), \\
      \sol' \gets \adversary^{H}(1^\secp,T(\secp),\ket{\puz})
    }\right]
    \le \frac{1}{2} + \negl(\secp).
    \]
\end{definition}

\begin{definition}[Efficiency]
\label{def:qtlp-efficiency}
A QTLP is \emph{efficient} if there exist polynomials $t_{gen}(\cdot)$ and $t_{sol}(\cdot)$ such that for all $\secp,T(\cdot)$ and every $\sol\in\{0,1\}$, the following bounds hold:
    \begin{enumerate}
        \item $\PGen^H(1^\secp,T,\sol)$ runs in time at most $t_{gen}(\secp,\log T)$.
        \item On every generated puzzle $\ket{\puz} \gets \PGen^{H}(1^\secp,T,\sol)$, $\PSolve^H(1^\secp,T,\ket{\puz})$ runs in time at most $t_{sol}(\secp,T)$ and uses at most $T$ rounds of queries.
    \end{enumerate}
\end{definition}

\subsection{BB84 Encodings from Binary Codes}

\begin{definition}[Binary linear codes]
An $[n,k,d]$-code $\mathcal{C}$ is a $k$-dimensional subspace of $\mathbb{F}^n_2$ such that 
\[
\min_{x,y \in \mathcal{C}: x \neq y} |x-y| \geq d.
\]
\end{definition}

\begin{definition}
\label{def:code-unitaries}
Let $\mathcal{C}$ be an $[n,k,d]$-code. For every codeword $x \in \mathcal{C}$, define the unitary $U_x := \bigotimes_{i=1}^n \Had^{x_i}$ where $x_i$ denotes the $i$-th bit of $x$ and $\Had$ is the Hadamard gate.
\end{definition}

If distinct codewords $c,c'\in\mathcal{C}$ differ in $t\geq d$ positions, then
\begin{equation}
  \left|\bra{x}U_c^\dagger U_{c'}\ket{x'}\right| \leq 2^{-t/2} \leq 2^{-d/2}.
  \label{eq:code-bb84-overlap}
\end{equation}
Each differing coordinate contributes a factor of magnitude $2^{-1/2}$, while an agreeing coordinate contributes either $0$ or $1$, which explains the first inequality. For the same codeword $c$, distinct computational-basis labels give orthogonal states. Thus, looking ahead, the family satisfies the overlap condition of \Cref{lem:hidden-pair-list} with $\varepsilon=2^{-d/2}$.

\subsection{Exact Grover Search}

We recall the exact Grover search of Brassard et al.~\cite[Section~2.1]{BHMT02}. Suppose exactly one of $N$ elements is marked, and let $\phi=\arcsin(1/\sqrt N)$. Ordinary Grover search succeeds in finding the marked element with probability $\sin^2((2q+1)\phi)$ after $q$ iterations. If $\tilde g=\pi/(4\phi)-1/2$ is an integer, this probability is one at $q=\tilde g$ since $(2\tilde{g}+1) \phi = \pi/2$. Otherwise, we round up to $g=\lceil\tilde g\rceil$ and reduce the initial marked amplitude to $\sin\theta$, where $\theta := \pi/(4g+2)$, so that $(2g+1)\theta = \pi/2$.  Alternating the reflection about the marked element $R_s=I-2\proj s$ with the reflection about the initial state $g$ times recovers the marked element exactly.

\begin{lemma}[\cite{BHMT02}]
\label{lem:exact-amplification}
Suppose exactly one element $s$ in $[N]$ is marked. Given oracle access to $R_s:=I-2\proj s$ on the space with orthonormal basis $\{\ket\bot,\ket1,\ldots,\ket N\}$, so that $R_s\ket\bot=\ket\bot$, there exists an algorithm that finds $s$ with probability $1$ by making $O(\sqrt N)$ queries to $R_s$.
\end{lemma}
For completeness, we provide a proof below.
\begin{proof}
If $N=1$, output the only element without making any queries. Otherwise, let $\phi:=\arcsin(1/\sqrt N)$, $g:=\lceil\pi/(4\phi)-1/2\rceil$, and $\theta:=\pi/(4g+2)$. Since $\theta\leq\phi$, the following state is normalized:
\begin{equation*}
  \ket{\mathsf{Init}}
  :=\sqrt{1-N\sin^2\theta}\,\ket\bot
    +\sin\theta\sum_{x\in[N]}\ket x.
\end{equation*}
Here $\ket\bot$ is orthogonal to every $\ket{x}$. Set $R_{\mathsf{Init}}:=2\proj{\mathsf{Init}}-I$ and write
\[
  \ket{\mathsf{Init}}=\sin\theta\ket s+\cos\theta\ket{\mathsf{Bad}},
  \qquad
  \ket{\mathsf{Bad}}:=\frac{\ket{\mathsf{Init}}-\sin\theta\ket s}{\cos\theta}.
\]
The states $\ket s$ and $\ket{\mathsf{Bad}}$ form an orthonormal basis $\mathcal B=(\ket s,\ket{\mathsf{Bad}})$ of the two-dimensional subspace $\mathcal K:=\operatorname{span}\{\ket s,\ket{\mathsf{Bad}}\}$. Both reflections ($R_s$ and $R_{\mathsf{Init}}$) preserve $\mathcal K$, since $\ket s$ and $\ket{\mathsf{Init}}$ belong to it. We now compute the matrices of these reflections on $\mathcal K$ with respect to the basis $\mathcal B$. Since $R_s$ maps $\ket s$ to $-\ket s$ and leaves $\ket{\mathsf{Bad}}$ unchanged, its matrix is
\[
  \begin{pmatrix}
    -1&0\\
    0&1
  \end{pmatrix}.
\]
The coordinate vector of $\ket{\mathsf{Init}}$ in this basis is $\binom{\sin\theta}{\cos\theta}$. Therefore, the matrix of $R_{\mathsf{Init}}=2\proj{\mathsf{Init}}-I$ is
\[
  \begin{aligned}
    2\begin{pmatrix}\sin\theta\\\cos\theta\end{pmatrix}
      \begin{pmatrix}\sin\theta&\cos\theta\end{pmatrix}-I_2
    &=\begin{pmatrix}
      2\sin^2\theta-1&2\sin\theta\cos\theta\\
      2\sin\theta\cos\theta&2\cos^2\theta-1
    \end{pmatrix}\\[2pt]
    &=\begin{pmatrix}
      -\cos(2\theta)&\sin(2\theta)\\
      \sin(2\theta)&\cos(2\theta)
    \end{pmatrix}.
  \end{aligned}
\]
Here we used $\sin(2\theta)=2\sin\theta\cos\theta$ and $\cos(2\theta)=\cos^2\theta-\sin^2\theta$. Since $R_s$ acts first in the product $R_{\mathsf{Init}}R_s$, the restriction of this product to $\mathcal K$ is represented, in the same basis, by the matrix product
\[  \begin{pmatrix}    -\cos(2\theta)&\sin(2\theta)\\
    \sin(2\theta)&\cos(2\theta)
  \end{pmatrix}
  \begin{pmatrix}
    -1&0\\
    0&1
  \end{pmatrix}
  =\begin{pmatrix}
    \cos(2\theta)&\sin(2\theta)\\
    -\sin(2\theta)&\cos(2\theta)
  \end{pmatrix}.
\]
The full operator acts as $-I$ on $\mathcal K^\perp$. Since $\ket{\mathsf{Init}}\in\mathcal K$, every iterate remains in $\mathcal K$, so its evolution is completely described by this $2\times2$ matrix.
Thus each iteration increases the angle from the bad direction by $2\theta$, and for every integer $q\geq0$,
\begin{equation}
  (R_{\mathsf{Init}}R_s)^q\ket{\mathsf{Init}}
  =\sin((2q+1)\theta)\ket s
    +\cos((2q+1)\theta)\ket{\mathsf{Bad}}.
  \label{eq:exact-grover-trajectory}
\end{equation}
At $q=g$, the right-hand side is $\ket s$ because $(2g+1)\theta=\pi/2$.
\end{proof}

\section{Construction}
\label{sec:construction}

We give our QTLP construction in the QROM. For the parameters chosen below, let $H_{i,\secp}:\{0,1\}^*\to\{0,1\}^{\secp}$ for $i\in[\secp]$, and $F:\{0,1\}^{*}\to\{0,1\}$  be independent random oracles.
\begin{construction} \label{con:lock} \hfill
\begin{itemize}
\item  $\PGen^H(1^\secp,T,\sol):$
    \begin{enumerate}
        \item Set $k=k(T):=2\lfloor\log_2(T+1)\rfloor-2$, $n:=\secp k$, and $d:=\secp$. Let $\mathcal{C}_{\secp,k}$ be an $[n,k,d]$-code (with an efficient encoding algorithm).\footnote{Such codes can be constructed easily using the binary linear code $\mathcal{C}_{\secp,k}$ with repetition encoder
        \[
          \mathcal{C}_{\secp,k}(a):=\underbrace{a\Vert a\Vert\cdots\Vert a}_{\secp\text{ copies}},
          \qquad a\in\{0,1\}^k.
        \]
        Since the Hamming distance between encoded messages is $\secp$ times their original distance, this is an $[n,k,d]=[\secp k,k,\secp]$-code.}
        \item For all $i \in[\secp]$, independently sample $a_i \gets \{0,1\}^{k}$ and $x_i \gets \{0,1\}^{n}$. 
        \item Prepare exactly one copy of each token
        \[ \ket{\psi_i} := U_{\mathcal{C}_{\secp,k}({a_i})} \ket{x_i}, \]
        where $\mathcal{C}_{\secp,k}(a_i)$ is the encoding of $a_i$ defined above, and the unitaries $U_x$ are defined in \Cref{def:code-unitaries}.
        \item Compute all tags $z_i$ and the mask $\msk$ using one round of $(\secp+1)$-parallel queries:
        \begin{equation*}
        z_i := H_{i,\secp}(a_i, x_i) \quad i \in[\secp],
        \qquad \msk := F(a_1,x_1,\ldots,a_\secp,x_\secp).
        \end{equation*}
        \item Set $c := \sol \oplus \msk$ and output
        \begin{equation*}
        \ket{\puz}
        := \bigotimes_{i \in [\secp]} \ket{z_i}_{\reg{Z}_i} \ket{\psi_i}_{\reg{B}_i} \otimes \ket{c}_{\reg{C}}.
        \end{equation*}
        Here $\reg{Z}_i$ holds the $\secp$-bit tag, $\reg{B}_i$ holds the $n$-qubit token, and $\reg{C}$ holds the masked bit.
    \end{enumerate}
\Needspace{6\baselineskip}
\item $\PSolve^H(1^\secp,T,\ket{\puz}):$
    \begin{enumerate}
        \item Read the classical tag registers $\reg{Z}_1,\ldots,\reg{Z}_\secp$ and the masked-bit register $\reg{C}$ to obtain $(z_1,\ldots,z_\secp,c)$, retaining the token registers $\reg{B}_1,\ldots,\reg{B}_\secp$. 
        \item Set $g := \lceil\frac{T}{2}\rceil - 1$ and $\theta := \frac{\pi}{4g + 2}$. For each $i  \in  [\secp]$, allocate a fresh search register $\reg{A}_i$, with basis labels $\{\bot\}\cup\{0,1\}^k$, and initialize it with the following state\footnote{For the specific choice of parameters, $2^k \sin^2(\theta) \leq 1$, which makes the state $\ket{\mathsf{Init}}$ well defined.}:
        \begin{equation}
            \ket{\mathsf{Init}} := \sqrt{1 - 2^{k} \sin^2 \theta}\, \ket{\bot} + \sin \theta \sum_{a \in \{0,1\}^{k}} \ket{a},
        \label{eq:biased-state}
        \end{equation}
        where $\bot$ is a dummy label and $\ket{\bot}$ is orthogonal to all $\ket a$. Let $R_{\mathsf{Init}}:=2\proj{\mathsf{Init}}-I$ act on this search register.
        \item For every $i\in[\secp]$ and candidate $a\in\{0,1\}^k$, define the projector on the token register $\reg{B}_i$ onto labels matching the tag by
        \[
          \Pi_{i,a,z_i}:=\sum_{x\in\{0,1\}^n:\,H_i(a,x)=z_i}\proj x.
        \]
        Define the reflection on $\reg{A}_i\reg{B}_i$ as
        \begin{equation}
          R_{i,z_i}:=\proj\bot_{\reg{A}_i}\otimes I_{\reg{B}_i}+
          \sum_{a\in\{0,1\}^k}\proj a_{\reg{A}_i}\otimes
          U_{\mathcal C_{\secp,k}(a)}(I_{\reg{B}_i}-2\Pi_{i,a,z_i}) U_{\mathcal C_{\secp,k}(a)}^\dagger.
          \label{eq:marking-operator}
        \end{equation}
        All $U$ and $\Pi$ operators in this expression act on $\reg{B}_i$. To implement it, use $\reg{A}_i$ as the control and apply the candidate-controlled $U_{\mathcal C_{\secp,k}(a)}^\dagger$, compute $H_i(a,x)$ coherently, flip the phase when the search label is not $\bot$ and the hash value equals $z_i$, uncompute the oracle answer, and apply the candidate-controlled $U_{\mathcal C_{\secp,k}(a)}$. This uses two rounds of oracle queries and leaves the dummy branch unchanged.
        \item For all $i\in[\secp]$, perform the following $g$ times in parallel: apply $R_{i,z_i}$ to $\reg{A}_i\reg{B}_i$, followed by $R_{\mathsf{Init}}$ to $\reg{A}_i$.
        \item For each $i$, measure $\reg{A}_i$ in the computational basis to obtain $a'_i$; abort if any outcome is $\bot$. Apply $U_{\mathcal C_{\secp,k}(a'_i)}^\dagger$ to $\reg{B}_i$, then measure $\reg{B}_i$ in the computational basis to obtain $x'_i$.
        \item Query $F(a'_1,x'_1,\ldots,a'_\secp,x'_\secp)$ to obtain $\mathsf{msk}'$. Output
        \[
        \sol' := c \oplus \mathsf{msk}'.
        \]
    \end{enumerate}
\end{itemize}
\end{construction}
By~\Cref{eq:code-bb84-overlap}, this choice gives $\varepsilon=2^{-\secp/2}$, so $p\varepsilon$ is negligible for every polynomial query width $p=p(\secp)$. Each token contains $n=\secp k$ qubits, and the $\secp$ tokens contain $\secp n=O(\secp^2\log T)$ qubits in total.

\begin{lemma}[Efficiency]
\label{lem:efficiency}
\Cref{con:lock} is efficient.
\end{lemma}
\begin{proof}
We count each one- or two-qubit gate and each oracle call as one operation, allowing single-qubit rotations with efficiently computable angles.\footnote{Over a fixed finite universal gate set, approximate each of the $M=\mathrm{poly}(\secp,T)$ non-oracle gates to operator-norm error at most $2^{-\secp}/(4M)$. This gives polynomial overhead and adds at most $2^{-\secp}$ to the failure probability, without changing oracle depth.} The generator prepares $\secp^2k$ BB84 qubits and makes $\secp+1$ classical oracle calls in one round, so its running time, including oracle calls, is $\mathrm{poly}(\secp,\log T)$.

Using a flag qubit to distinguish $\bot$ from the candidate labels, the solver prepares each $\ket{\mathsf{Init}}$ with one single-qubit rotation and $k$ controlled Hadamard gates. Reflection about this state uses the preparation circuit, its inverse, and a reflection about the all-zero state, requiring $\mathrm{poly}(k)$ gates. Each search iteration uses $\mathrm{poly}(\secp,\log T)$ non-oracle gates and $2\secp$ oracle calls across all $\secp$ copies, in two parallel rounds. Thus, $g=\lceil T/2\rceil-1$ iterations and the final call to $F$ use at most $2g\secp+1=O(\secp T)$ oracle calls in $2g+1\leq T$ rounds, with at most $\secp$ calls per round. The total running time, including oracle calls, is $\mathrm{poly}(\secp,T)$.
\end{proof}

\section{Proof of Correctness}

\begin{lemma}[Correctness] \label{lem:correctness}
\Cref{con:lock} is correct.
\end{lemma}
\begin{proof}
Fix $\secp$, $T$, and $\sol$.\\

\noindent \textbf{Registers.}
We use the register names from \Cref{con:lock}. For each $i\in[\secp]$, the register $\reg{Z}_i$ contains the classical tag $z_i$, and $\reg{B}_i$ contains the token $\ket{\psi_i}=U_{\mathcal C_{\secp,k}(a_i)}\ket{x_i}$ corresponding to the hidden pair $(a_i,x_i)$ chosen by the generator. The register $\reg{C}$ contains the masked bit $c$, and $\reg{A}_i$ is the solver's fresh search register, with basis labels $\{\bot\}\cup\{0,1\}^k$. After reading the tags and $c$, we suppress $\reg{Z}_i$ and $\reg{C}$ in the analysis. All search states $\ket a$, $\ket\bot$, $\ket{\mathsf{Init}}$, and $\ket{\mathsf{Bad}_i}$ below belong to $\reg{A}_i$; all token states $\ket x$ and $\ket{\psi_i}$, and all $U$ and $\Pi$ operators, belong to $\reg{B}_i$. Each tensor product is ordered as $\reg{A}_i\reg{B}_i$.\\

\noindent \textbf{The ideal solver.}
For the analysis, define an ideal version $\widetilde{\PSolve}$ of the solver. In addition to the puzzle, this auxiliary algorithm has oracle access, for each $i\in[\secp]$, to the reflection
\[
  R_i:=I_{\reg{A}_i}-2\proj{a_i}_{\reg{A}_i}.
\]
This reflection changes the sign of $\ket{a_i}$ and leaves every other basis state of $\reg{A}_i$, including $\ket\bot$, unchanged. On the joint registers it acts as $R_i\otimes I_{\reg{B}_i}$, so it does not affect the token. Access to these reflections, which depend on the hidden labels $a_i$, is used only for comparison in the proof. The construction implements the reflections $R_{i,z_i}$ instead. Throughout the following analysis, both solvers use exact state preparations and gates; the error from finite-gate approximation is included at the end.

The ideal solver proceeds as follows:
\begin{enumerate}
  \item Read the tags $z_1,\ldots,z_\secp$ and the masked bit $c$, retaining the token registers. Use the same parameters $g=\lceil T/2\rceil-1$ and $\theta=\pi/(4g+2)$ as $\PSolve$, and initialize each $\reg{A}_i$ in the state $\ket{\mathsf{Init}}$ from \Cref{eq:biased-state}. The joint state of $\reg{A}_i\reg{B}_i$ is then
  \[
    \ket{\Phi_i}:=\ket{\mathsf{Init}}_{\reg{A}_i}\otimes\ket{\psi_i}_{\reg{B}_i}.
  \]
  \item For all $i\in[\secp]$, perform $g$ search iterations in parallel. In each iteration, apply $R_i$ to $\reg{A}_i$, followed by $R_{\mathsf{Init}}$ to the same register. Thus one ideal iteration on $\reg{A}_i\reg{B}_i$ is
  \[
    \widetilde{G}_i:=(R_{\mathsf{Init}}R_i)\otimes I_{\reg{B}_i}.
  \]
  In particular, the token register remains in $\ket{\psi_i}$ throughout the ideal search.
  \item Measure each $\reg{A}_i$ in the computational basis to obtain $a'_i$, aborting if any outcome is $\bot$. For each $i$, apply $U_{\mathcal C_{\secp,k}(a'_i)}^\dagger$ to $\reg{B}_i$ and measure it in the computational basis to obtain $x'_i$.
  \item Query $F(a'_1,x'_1,\ldots,a'_\secp,x'_\secp)$ to obtain $\mathsf{msk}'$, and output $\sol'=c\oplus\mathsf{msk}'$.
\end{enumerate}
Thus the ideal solver uses the same initialization and final decoding and output procedure as $\PSolve$; its search uses $R_i\otimes I_{\reg{B}_i}$ in place of $R_{i,z_i}$.\\

\noindent \textbf{Correctness of the ideal solver.}
If $g=0$, then $k=0$ and $n=0$, so there is only one possible pair $(a_i,x_i)$ for each $i$. Both solvers recover the correct mask with their final query to $F$, and correctness is perfect in this case. Henceforth, assume $g\geq1$.
By \Cref{eq:exact-grover-trajectory}, the ideal search satisfies
\[
  \widetilde{G}_i^g\ket{\Phi_i}
  =\ket{a_i}_{\reg{A}_i}\otimes\ket{\psi_i}_{\reg{B}_i}.
\]
Indeed, the search register has the initialization used in that equation with marked label $a_i$, and $(2g+1)\theta=\pi/2$. Measuring $\reg{A}_i$ therefore returns $a_i$ with certainty. The subsequent decoding maps $\ket{\psi_i}$ to $\ket{x_i}$, so measuring $\reg{B}_i$ returns $x_i$ with certainty as well. The ideal solver consequently queries $F$ on the generator's tuple, recovers $\msk$, and outputs $c\oplus\msk=\sol$.\\

\noindent \textbf{Comparison with the actual solver.}
The actual solver uses the reflection $R_{i,z_i}$ from \Cref{eq:marking-operator}, so its Grover iteration on $\reg{A}_i\reg{B}_i$ is
\[
  G_i:=(R_{\mathsf{Init}}\otimes I_{\reg{B}_i})R_{i,z_i}.
\]
Starting from the same state $\ket{\Phi_i}$, the actual and ideal searches therefore produce $G_i^g\ket{\Phi_i}$ and $\widetilde{G}_i^g\ket{\Phi_i}$, respectively. For fixed generation data and oracles, we first bound the distance between these states and then use this bound to control the actual solver's failure probability.

\begin{claim}
\label{clm:solver-distance}
For every  $i \in [\secp]$, define the quantity
\[
  W_i:=\sum_{\substack{a\in\{0,1\}^k\\a\ne a_i}}
  \left\|\Pi_{i,a,z_i}U_{\mathcal C_{\secp,k}(a)}^\dagger\ket{\psi_i}\right\|^2.
\]
Then
\begin{equation*}
  \left\|(G_i^g-\widetilde{G}_i^g)
  \ket{\Phi_i}\right\|
  \leq\sqrt{W_i}.
\end{equation*}
\end{claim}
\begin{proof}
We prove the claim by a hybrid argument. Consider the telescoping identity
\[
  G_i^g - \widetilde{G}_i^g
  =\sum_{q=0}^{g-1} G_i^{g-1-q} (G_i - \widetilde{G}_i) \widetilde{G}_i^q.
\]
To verify this identity, expand the difference in each summand. For $0\leq q\leq g-1$,
\[
  G_i^{g-1-q}(G_i-\widetilde{G}_i)\widetilde{G}_i^q
  =G_i^{g-q}\widetilde{G}_i^q
   -G_i^{g-1-q}\widetilde{G}_i^{q+1}.
\]
Summing these differences and reindexing the second sum gives
\begin{align*}
  \sum_{q=0}^{g-1}G_i^{g-1-q}(G_i-\widetilde{G}_i)\widetilde{G}_i^q
  &=\sum_{q=0}^{g-1}G_i^{g-q}\widetilde{G}_i^q
    -\sum_{q=1}^{g}G_i^{g-q}\widetilde{G}_i^q\\
  &=G_i^g-\widetilde{G}_i^g.
\end{align*}
Every term with index $1\leq q\leq g-1$ appears once with each sign and cancels. The only surviving terms are the $q=0$ term $G_i^g$ from the first sum and the $q=g$ term $-\widetilde{G}_i^g$ from the second sum. The order of the operators is preserved throughout; no commutativity assumption is needed.

By the above identity, the triangle inequality, and unitary invariance of the Euclidean norm,
\begin{align*}
  \left \|(G_i^g - \widetilde{G}_i^g) \ket{\Phi_i} \right\|
  &\leq \sum_{q=0}^{g-1} \left\|(G_i - \widetilde{G}_i) \widetilde{G}_i^q \ket{\Phi_i} \right\|\\
  &=\sum_{q=0}^{g-1} \left\|
  (R_{i,z_i} - R_i\otimes I_{\reg{B}_i}) \widetilde{G}_i^q \ket{\Phi_i} \right\|.
\end{align*}
We now evaluate each summand. Define $\ket{\mathsf{Bad}_i}:=(\ket{\mathsf{Init}}-\sin\theta\ket{a_i})/\cos\theta$. This is a unit vector orthogonal to $\ket{a_i}$. For $0\leq q\leq g-1$, we have
\[
  \widetilde{G}_i^q\ket{\Phi_i}
  =\bigl(\sin((2q+1)\theta)\ket{a_i}
  +\cos((2q+1)\theta)\ket{\mathsf{Bad}_i}\bigr)\ket{\psi_i}.
\]
Thus, after $q$ ideal Grover iterations, the amplitude of each incorrect search label $\ket a$, with $a\in\{0,1\}^k\setminus\{a_i\}$, in register $\reg{A}_i$ is
\[
\alpha_q := \cos((2q+1)\theta) \cdot \frac{\sin\theta}{\cos\theta}
= \cos((2q+1)\theta) \tan\theta \geq 0.
\]Next, the two reflections agree on every state whose search register has label $\bot$: for any token state $\ket v_{\reg{B}_i}$,
\[
  R_{i,z_i}(\ket\bot_{\reg{A}_i}\otimes\ket v_{\reg{B}_i})
  =(R_i\otimes I_{\reg{B}_i})(\ket\bot_{\reg{A}_i}\otimes\ket v_{\reg{B}_i})
  =\ket\bot_{\reg{A}_i}\otimes\ket v_{\reg{B}_i}.
\]
Here $R_{i,z_i}$ acts as the identity on this subspace by its definition in \Cref{eq:marking-operator}, while $R_i\ket\bot=\ket\bot$ because $\langle a_i|\bot\rangle=0$.

 On the particular joint state $\ket{a_i}_{\reg{A}_i}\otimes\ket{\psi_i}_{\reg{B}_i}$, both reflections apply a phase of $-1$. Explicitly,
\[
  R_{i,z_i}(\ket{a_i}_{\reg{A}_i}\otimes\ket{\psi_i}_{\reg{B}_i})
  =(R_i\otimes I_{\reg{B}_i})(\ket{a_i}_{\reg{A}_i}\otimes\ket{\psi_i}_{\reg{B}_i})
  =-\ket{a_i}_{\reg{A}_i}\otimes\ket{\psi_i}_{\reg{B}_i}.
\]
For the ideal reflection, $R_i\ket{a_i}=-\ket{a_i}$ by definition. For the actual reflection, decoding gives $U_{\mathcal C_{\secp,k}(a_i)}^\dagger\ket{\psi_i}=\ket{x_i}$, and $\Pi_{i,a_i,z_i}\ket{x_i}=\ket{x_i}$ because $H_i(a_i,x_i)=z_i$. Thus $I-2\Pi_{i,a_i,z_i}$ sends $\ket{x_i}$ to $-\ket{x_i}$, and re-encoding gives $-\ket{\psi_i}$, while the search register remains in $\ket{a_i}$.

 For an incorrect candidate $a\ne a_i$, their difference on the token $\ket{\psi_i}$ is
\[
  U_{\mathcal C_{\secp,k}(a)}(I-2\Pi_{i,a,z_i})
  U_{\mathcal C_{\secp,k}(a)}^\dagger \ket{\psi_i} - \ket{\psi_i}
  = -2 U_{\mathcal C_{\secp,k}(a)}\Pi_{i,a,z_i}
  U_{\mathcal C_{\secp,k}(a)}^\dagger\ket{\psi_i}.
\]
Putting things together, we have
\begin{align*}
& \big\|(R_{i,z_i} - R_i \otimes I_{\reg{B}_i})
  \widetilde{G}_i^q \ket{\Phi_i} \big\| \\
  = & \bigg\|\sum_{a \in \{0,1\}^k:a \neq a_i}  -2 \alpha_q \ket{a}_{\reg{A}_i} \otimes U_{\mathcal C_{\secp,k}(a)}\Pi_{i,a,z_i}
  U_{\mathcal C_{\secp,k}(a)}^\dagger\ket{\psi_i} \bigg\| \\
  = & 2 \alpha_q \sqrt{ \sum_{a \in \{0,1\}^k:a \neq a_i}  \bigg\|U_{\mathcal C_{\secp,k}(a)}\Pi_{i,a,z_i} U_{\mathcal C_{\secp,k}(a)}^\dagger\ket{\psi_i} \bigg\|^2} \\
  = & 2 \tan\theta \cos((2q+1)\theta) \sqrt{W_i}.
\end{align*}
Summing up the errors,
\[
  \left\|(G_i^g-\widetilde{G}_i^g)\ket{\Phi_i}\right\|
  \leq 2 \sqrt{W_i}\tan\theta
  \sum_{q=0}^{g-1}\cos((2q+1)\theta).
\]
Finally, $(2g+1)\theta=\pi/2$ implies
\[
  \sum_{q=0}^{g-1}\cos((2q+1)\theta)
  =\frac{\sin(2g\theta)}{2\sin\theta}
  =\frac{\cos\theta}{2\sin\theta}
  = \frac{\cot\theta}{2}.
\]
Substituting into the above bound and using $\tan\theta\cot\theta=1$ gives the claimed upper bound $\sqrt{W_i}$.
\end{proof}
To relate this state-distance bound to the solver's output, let
\[
  P_i:=I_{\reg{A}_i\reg{B}_i}
       -\proj{a_i}_{\reg{A}_i}\otimes\proj{\psi_i}_{\reg{B}_i}.
\]
This is the failure projector for recovering the pair $(a_i,x_i)$, expressed before the final decoding and measurement. Indeed, the solver first measures $\reg{A}_i$, then applies $U_{\mathcal C_{\secp,k}(a'_i)}^\dagger$ to $\reg{B}_i$ and measures it; the outcome $(a_i,x_i)$ therefore corresponds to the rank-one projector onto $\ket{a_i}_{\reg{A}_i}\ket{\psi_i}_{\reg{B}_i}$. In particular, $P_i$ also counts an outcome of $\bot$ as failure. The ideal final state is annihilated by $P_i$, so
\begin{align*}
\Pr[(a'_i,x'_i)\ne(a_i,x_i)]
&=\|P_iG_i^g\ket{\Phi_i}\|^2\\
&=\|P_i(G_i^g-\widetilde{G}_i^g)\ket{\Phi_i}\|^2\\
&\leq\|(G_i^g-\widetilde{G}_i^g)\ket{\Phi_i}\|^2
\leq W_i,
\end{align*}
where this probability is over the solver's measurements for fixed generation data and oracles. We used that $P_i$ is a contraction and then applied \Cref{clm:solver-distance}.

It remains to bound the expectation of $W_i$. Condition on $a_i,x_i,z_i$. For any $a\ne a_i$ and $x\in\{0,1\}^n$, the value $H_i(a,x)$ remains uniform on $\{0,1\}^\secp$. Therefore, 
\begin{align*}
& \E_{H_i}[\Pi_{i,a,z_i} | H_i(a_i,x_i) = z_i ] \\
= & \E_{H_i}\left[ \sum_{x:H_i(a,x) = z_i} \proj{x} \middle| \, H_i(a_i,x_i) = z_i \right] \\
= & \E_{H_i}\left[ \sum_x \mathbf{I}[H_i(a,x) = z_i] \proj{x} \middle| H_i(a_i,x_i) = z_i \right] \\
= & \sum_x \Pr_{H_i}\left[H_i(a,x) = z_i \mid H_i(a_i,x_i) = z_i\right] \proj{x} \\
= & \sum_{x} 2^{-\secp} \proj{x} \\
= & 2^{-\secp} I_{\reg{B}_i}.
\end{align*}
Since $\Pi_{i,a,z_i}$ is a projector and $U_{\mathcal C_{\secp,k}(a)}^\dagger\ket{\psi_i}$ is a unit vector determined by $(a_i,x_i)$, each summand in $W_i$ has conditional expectation $2^{-\secp}$. There are $2^k-1$ incorrect candidates $a\ne a_i$, so averaging also over the generation data gives
\[
  \E[W_i]=\frac{2^k-1}{2^\secp}.
\]
Combining this with the measurement bound above, the probability that the solver fails to recover the pair $(a_i,x_i)$, now over the random oracles, puzzle generation, and the solver's measurements, is at most $(2^k-1)/2^\secp$. A union bound over all $\secp$ tokens therefore gives
\[
  \Pr[\text{some pair }(a'_i,x'_i)\ne(a_i,x_i)]
  \leq\frac{\secp(2^k-1)}{2^\secp}.
\]
If every pair is recovered correctly, the solver queries $F$ on exactly the tuple used by the generator. Consequently, $\mathsf{msk}'=\msk$ and the output is $\sol'=c\oplus\mathsf{msk}'=(\sol\oplus\msk)\oplus\msk=\sol$. Thus,
\[
  \Pr[\sol'=\sol]\geq1-\frac{\secp(2^k-1)}{2^\secp}.
\]
Here an abort is counted as failure. For every polynomially bounded delay $T=T(\secp)$, we have $2^k\leq(T+1)^2/4$, so the numerator $\secp(2^k-1)$ is polynomial in $\secp$, whereas the denominator is $2^\secp$. The failure probability of the actual solver with exact gates is therefore negligible in $\secp$. By \Cref{lem:efficiency}, the finite-gate implementation adds at most $2^{-\secp}$ to this failure probability, so its success probability is at least
\[
  1-\frac{\secp(2^k-1)}{2^\secp}-2^{-\secp}
  =1-\mathrm{negl}(\secp).
\]
This holds for every solution bit $\sol$, as required for correctness.
\end{proof}

\section{Proof of Security}
\label{sec:security-proof}

\subsection{Useful Lemmata}

\begin{lemma}[{\cite[Lemma~2]{TFKW13}}]
\label{lem:tfkw-operator-sum}
Let \(A_1,\ldots,A_m\allowbreak\in\mathcal{P}(\mathcal{H})\)\footnote{$\mathcal{P}(\mathcal{H})$ is the set of all positive semi-definite operators on $\mathcal{H}$.}, and let \(\{\pi^k\}_{k=1}^{m}\) be a family of mutually orthogonal permutations of \([m]\); namely, for every \(i\in[m]\) and every distinct \(k,k'\in[m]\),
\[
    \pi^k(i)\neq \pi^{k'}(i).
\]
Then
\[
    \left\|
        \sum_{i=1}^{m} A_i
    \right\|_{\infty}
    \leq
    \sum_{k=1}^{m}
    \max_{i\in[m]}
    \left\|
        \sqrt{A_i}\sqrt{A_{\pi^k(i)}}
    \right\|_{\infty}.
\]
\end{lemma}

\begin{corollary}
\label{cor:projector-sum-coherence}
Let \(P_i=\lvert\psi_i\rangle\!\langle\psi_i\rvert\), for
\(i\in[m]\), be rank-one projectors satisfying
\[
    \max_{i\neq j}
    \bigl|\langle\psi_i\mid\psi_j\rangle\bigr|
    \leq \varepsilon.
\]
Then
\[
    \left\|
        \sum_{i=1}^{m}P_i
    \right\|_{\infty}
    \leq 1+(m-1) \varepsilon.
\]
\end{corollary}

\begin{proof}
For \(k\in\{0,\ldots,m-1\}\), let
\[
    \pi^k(i)=1+((i-1+k)\bmod m)
\]
denote the cyclic shifts of \([m]\). These permutations are mutually
orthogonal. Applying Lemma~\ref{lem:tfkw-operator-sum} and using
\(\sqrt{P_i}=P_i\), we obtain
\[
\begin{aligned}
    \left\|\sum_{i=1}^{m}P_i\right\|_{\infty}
    &\leq
    \sum_{k=0}^{m-1}
    \max_{i\in[m]}
    \left\|P_iP_{\pi^k(i)}\right\|_{\infty} \\
    &\leq 1+(m-1) \varepsilon,
\end{aligned}
\]
where we used
\[
    \|P_iP_j\|_{\infty}
    =
    \bigl|\langle\psi_i\mid\psi_j\rangle\bigr|. \qedhere
\]
\end{proof}
The following lemma is crucial in our security proof.
\begin{lemma}[List-one-wayness lemma]
\label{lem:hidden-pair-list}
Let $\{U_a\}_{a\in\{0,1\}^k} \subset \mathrm{U}(2^n)$ be a set of unitaries satisfying 
\[
|\bra{x} U^\dagger_a U_{a'} \ket{x'}| \leq \varepsilon
\]
for any distinct pairs $(a,x), (a',x') \in \{0,1\}^k \times \{0,1\}^n$. Let $p > 0$ be an integer and $\varepsilon\geq0$. Consider the following experiment:
\begin{enumerate}
    \item The challenger independently samples $a \gets \{0,1\}^k$ and $x \gets \{0,1\}^n$.
    \item The challenger sends one copy of $U_a \ket{x}$ to the adversary which is computationally unbounded.
    \item The adversary outputs a list $L\subseteq\{0,1\}^k\times\{0,1\}^n$ of $p$ distinct pairs.
    \item The adversary wins if $(a,x)\in L$.
\end{enumerate}
Then the winning probability of any adversary is at most $\frac{1+ (p-1)\varepsilon}{2^k}$.
\end{lemma}

\begin{proof}
Let
\[
    \mathcal{I}
    :=
    \{0,1\}^k\times\{0,1\}^n
\]
be the set of possible hidden pairs, and define
\[
    \ket{\psi_{a,x}}:=U_a\ket{x},
    \qquad
    P_{a,x}:=
    \ket{\psi_{a,x}}\!\bra{\psi_{a,x}}.
\]
The premise says that for any two distinct pairs $(a,x)\neq(a',x')$, we have
\[
    \left|
        \braket{\psi_{a,x}}{\psi_{a',x'}}
    \right|
    \leq \varepsilon.
\]
Repeated pairs do not increase the adversary's winning probability, so we represent its output as a set $L\subseteq\mathcal I$ with $|L|\leq p$. The adversary's strategy is described by a POVM:
\[
    \{M_L\}_{L\in\mathcal{L}_p},
    \qquad
    \mathcal{L}_p
    :=
    \{L\subseteq\mathcal{I}:|L|\leq p\},
\]
where
\[
    M_L \geq 0
    \qquad \text{and} \qquad
    \sum_{L\in\mathcal{L}_p} M_L = I_{2^n}.
\]
Fix a nonempty list of size $m:=|L|\leq p$,
\[
    L=\{(a_1,x_1),\ldots,(a_m,x_m)\},
\]
and define
\[
    A_L:=\sum_{i=1}^{m}P_{a_i,x_i}.
\]
For $t\in\{0,\ldots,m-1\}$, let $\pi^t$ be the cyclic permutation
\[
    \pi^t(i):=1+\bigl((i-1+t)\bmod m\bigr).
\]
The permutations $\{\pi^t\}_{t=0}^{m-1}$ are mutually orthogonal.
Therefore, Lemma~\ref{lem:tfkw-operator-sum} gives
\begin{align*}
    \|A_L\|_\infty
    &\leq
    \sum_{t=0}^{m-1}
    \max_{i\in[m]}
    \left\|
        P_{a_i,x_i}
        P_{a_{\pi^t(i)},x_{\pi^t(i)}}
    \right\|_\infty.
\end{align*}
The term corresponding to $t=0$ is equal to $1$. For every
$t\neq 0$, the two pairs are distinct, and hence
\begin{align*}
    \left\|
        P_{a_i,x_i}
        P_{a_{\pi^t(i)},x_{\pi^t(i)}}
    \right\|_\infty
    &=
    \left|
        \braket{\psi_{a_i,x_i}}
                {\psi_{a_{\pi^t(i)},x_{\pi^t(i)}}}
    \right|
    \leq \varepsilon.
\end{align*}
It follows that
\begin{equation}
    \|A_L\|_\infty
    \leq
    1+(m-1)\varepsilon\leq 1+(p-1)\varepsilon.
    \label{eq:hidden-pair-list-operator-bound}
\end{equation}
For $L=\varnothing$, set $A_L=0$; the same upper bound holds.
The adversary's winning probability therefore satisfies
\begin{align*}
    \Pr[\mathrm{win}]
    &=
    \frac{1}{2^{k+n}}
    \sum_{(a,x)\in \mathcal{I}}
    \sum_{\substack{
        L\in\mathcal{L}_p: \\
    (a,x) \in L}}
    \operatorname{Tr}\!\left(M_LP_{a,x}\right) \\
    &=
    \frac{1}{2^{k+n}}
    \sum_{L\in\mathcal{L}_p}
    \sum_{(a,x)\in L}
    \operatorname{Tr}\!\left(M_LP_{a,x}\right) \tag{by double counting} \\
    &=
    \frac{1}{2^{k+n}}
    \sum_{L\in\mathcal{L}_p}
    \operatorname{Tr}\!\left(M_LA_L\right) \tag{by definition of $A_L$} \\
    &\leq
    \frac{1}{2^{k+n}}
    \sum_{L\in\mathcal{L}_p}
    \|A_L\|_\infty \operatorname{Tr}(M_L) \tag{by $A_L \leq \|A_L\|_\infty I_{2^n}$ and $M_L \geq 0$} \\
    &\leq
    \frac{1+(p-1)\varepsilon}{2^{k+n}}
    \sum_{L\in\mathcal{L}_p}
    \operatorname{Tr}(M_L) \tag{by~\Cref{eq:hidden-pair-list-operator-bound}} \\
    &=
    \frac{1+(p-1)\varepsilon}{2^{k+n}}
    \operatorname{Tr}(I_{2^n}) \\
    &=
    \frac{1+(p-1)\varepsilon}{2^k}. \qedhere
\end{align*}
\end{proof}

Next, we consider a variant of~\Cref{lem:hidden-pair-list}  in the QROM. The differences from the experiment in~\Cref{lem:hidden-pair-list} are highlighted in red.
\begin{lemma} \label{lem:hidden-pair-list-oracle}
Let $\{U_a\}_{a\in\{0,1\}^k} \subset \mathrm{U}(2^n)$ be a set of unitaries satisfying 
\[
|\bra{x} U^\dagger_a U_{a'} \ket{x'}| \leq \varepsilon
\]
for any distinct pairs $(a,x), (a',x') \in \{0,1\}^k \times \{0,1\}^n$. Let  $p>0$\textcolor{red}{, $r\geq0$} be integers. Consider the following experiment $G$:
\begin{enumerate}
    \item \textcolor{red}{Let $H:\{0,1\}^k \times \{0,1\}^n \to \{0,1\}^\secp$ be a random oracle.}
    \item The challenger independently samples $a \gets \{0,1\}^k$ and $x \gets \{0,1\}^n$. \textcolor{red}{Set $z := H(a,x)$}.
    \item The challenger sends one copy of $U_a \ket{x}$ \textcolor{red}{and $z$} to the adversary, which is computationally unbounded.
    \item \textcolor{red}{The adversary makes $r$ rounds of $p$-parallel queries to $H$.}
    \item \textcolor{red}{The adversary outputs $(a',x') \in \{0,1\}^k \times \{0,1\}^n$.}
    \item \textcolor{red}{The experiment outputs $1$ if $(a',x') = (a,x)$.}
\end{enumerate}
Then
\begin{equation}
  \Pr[G=1]
  \leq
  \frac{1}{2^k}
  +2r\sqrt{\frac{1+(p-1)\varepsilon}{2^k}}.
  \label{eq:oracle-list-bound}
\end{equation}
\end{lemma}

\begin{proof}
We prove the lemma by a standard QROM reprogramming argument. Fix an adversary $A$. Consider the following sequence of hybrids, which gradually decouples the tag $z$ from the oracle $H$. The differences are highlighted in red:

\noindent $\Hyb(j)$ for $j \in \{1,\ldots,r+1\}$:
\begin{enumerate}
\item Let $H:\{0,1\}^k \times \{0,1\}^n \to \{0,1\}^\secp$ be a random oracle.
\item The challenger samples independently $a \gets \{0,1\}^k$, $x \gets \{0,1\}^n$, and \textcolor{red}{$z \gets \{0,1\}^\secp$}.
\item The challenger sends one copy of $U_a \ket{x}$ and $z$ to the adversary $A$.
\item The adversary $A$ makes $r$ rounds of $p$-parallel queries to $H$ in the following way: \textcolor{red}{right before the $j$-th round, the challenger reprograms $H$ at $(a,x)$ to $z$.}
\item The adversary outputs $(a',x') \in \{0,1\}^k \times \{0,1\}^n$.
\item The game outputs $1$ if $(a',x') = (a,x)$.
\end{enumerate}
Informally, the proof follows this chain:
\[
  \begin{gathered}
    G \equiv \Hyb(1)\ \approx \ \Hyb(2)\ \approx \ \cdots\ \approx \ \Hyb(r+1), \qquad
    \Pr[\Hyb(r+1)=1]\leq\frac{1}{2^k}.
  \end{gathered}
\]
The experiment $G$ is identically distributed to $\Hyb(1)$. We next bound each adjacent hop. The only difference between $\Hyb(j)$ and $\Hyb(j+1)$ is whether the $j$-th round is answered by an oracle reprogrammed at $(a,x)$. Define $\delta^2$ as the probability that measuring all query addresses immediately before this round yields a list containing $(a,x)$, averaged over the randomness of the experiment. The two parallel query operators agree outside this event, so the standard hybrid argument~\cite{BBBV97,AHU19} gives
\[
|\Pr[\Hyb(j)=1] - \Pr[\Hyb(j+1)=1]| \leq 2\delta.
\]
Before round $j$, the oracle has not yet been reprogrammed. A computationally unbounded adversary in the experiment of \Cref{lem:hidden-pair-list} can therefore sample $H$ and $z$ independently, simulate the first $j-1$ rounds, and measure the next query addresses to obtain a list $L$ of at most $p$ pairs. By \Cref{lem:hidden-pair-list},
\[
  \delta^2=\Pr[(a,x)\in L]\leq\frac{1+(p-1)\varepsilon}{2^k}.
\]
Consequently, each adjacent hop has error at most $2\sqrt{(1+(p-1)\varepsilon)/2^k}$.

In $\Hyb(r+1)$, the oracle is never reprogrammed, so $H$ and $z$ are independent of $(a,x)$. The same simulation reduces this terminal experiment to \Cref{lem:hidden-pair-list} with $p=1$, giving $\Pr[\Hyb(r+1)=1]\leq 1/2^k$. Summing the $r$ adjacent-hop bounds proves \Cref{eq:oracle-list-bound}.
\end{proof}

By parallel repetition~\cite[Theorem~4.9]{Gut10}, we obtain the following consequence of \Cref{lem:hidden-pair-list-oracle}.
\begin{corollary}
\label{cor:hidden-pair-parallel-repetition}
Under the same assumption on the set of unitaries as in~\Cref{lem:hidden-pair-list-oracle}, consider $\secp$ independent copies of its experiment, with independent $(a_i,x_i)$ and random oracles $H_i$. A computationally unbounded adversary receives all challenge states and tags and makes $r$ rounds of at most $p$ parallel queries to each $H_i$, allowing arbitrary joint quantum computation across the copies. If it outputs a candidate pair $(a_i',x_i')$ for each $i\in[\secp]$, then
\begin{equation*}
  \Pr\!\left[(a_i',x_i')=(a_i,x_i)\text{ for every }i\in[\secp]\right]
  \leq \beta_{r,p,\varepsilon,N}^{\secp},
\end{equation*}
where $\beta_{r,p,\varepsilon,N}$ denotes the right-hand side of \Cref{eq:oracle-list-bound}.
\end{corollary}
\begin{proof}
View one copy as an interactive verifier that privately samples $(a,x)$ and the complete finite table of $H$, sends $U_a\ket{x}$ and $z=H(a,x)$, and answers $r$ batches of at most $p$ parallel oracle queries. It accepts precisely when the prover's final pair equals $(a,x)$. Its optimal acceptance probability $v$ is at most $\beta_{r,p,\varepsilon,N}$ by~\Cref{lem:hidden-pair-list-oracle}.

The repeated verifier consists of $\secp$ independent copies and accepts only if every copy accepts. Its accepting verifier operator is therefore the tensor product of the single-copy accepting operators. Gutoski's all-success parallel-repetition theorem~\cite[Theorem~4.9]{Gut10} bounds its acceptance probability by $v^{\secp}\leq\beta_{r,p,\varepsilon,N}^{\secp}$, even for an arbitrary joint quantum prover entangling all copies.
\end{proof}

\subsection{Security Proof}

We prove security. The proof is very similar to that of~\Cref{lem:hidden-pair-list-oracle}.

\begin{lemma}[Security] \label{lem:security}
Let $T=T(\secp)\geq3$ be a polynomially bounded integer delay, let $p=p(\secp)\geq1$ be a polynomially bounded integer query width, and let $r=r(\secp)=o(T)$ be a nonnegative integer number of rounds. Any adversary making $r$ rounds of $p$-parallel queries against \Cref{con:lock} satisfies
\begin{equation}
  \Pr[\sol'=\sol]\leq\frac12+2r\sqrt{p\beta_{r,p,\varepsilon,N}^{\secp}}=\frac12 + \mathrm{negl}(\secp),
  \label{eq:security-bound}
\end{equation}
where $\beta_{r,p,\varepsilon,N}$ is as in \Cref{cor:hidden-pair-parallel-repetition} with $\varepsilon=2^{-\secp/2}$.
\end{lemma}

\begin{proof}
We introduce a sequence of hybrids, which gradually decouples the mask $\msk$ from the oracle $F$.

\noindent $\Hyb(j)$ for $j \in \{1,\ldots,r+1\}$:
\begin{enumerate}
\item $H_i$ for $i\in[\secp]$ and $F$ are independent random oracles.
\item The challenger samples $\sol \gets \{0,1\}$.
\item For every $i \in [\secp]$, the challenger samples $a_i \gets \{0,1\}^k$ and $x_i \gets \{0,1\}^n$, and sets $\ket{\psi_i} = U_{\mathcal{C}_{\secp,k}(a_i)}\ket{x_i}$ and $z_i = H_i(a_i,x_i)$.
\item Set $\msk \gets \{0,1\}$ and $\ket{\puz} = \bigotimes_{i\in[\secp]}\ket{z_i} \ket{\psi_i} \otimes \ket{\sol \oplus \msk}$.
\item The challenger sends $\ket{\puz}$ to the adversary $A$. 
\item The adversary $A$ makes $r$ rounds of $p$-parallel queries in the following way: right before the $j$-th round of queries, the challenger reprograms $F$ at $(a_1,x_1,\allowbreak\ldots,\allowbreak a_\secp,x_\secp)$ to $\msk$.
\item The adversary $A$ outputs $\sol' \in \{0,1\}$.
\item Output $1$ if and only if $\sol'=\sol$.
\end{enumerate}
Informally, the proof follows this chain:
\begin{equation}
  \Real \equiv\ \Hyb(1)
  \ \approx\  \Hyb(2)\ \approx\ \cdots\ \approx\  \Hyb(r+1), \quad \Pr[\Hyb(r+1) = 1] = \frac{1}{2}.
  \label{eq:main-game-chain}
\end{equation}
Immediately before round $j$, measure the addresses of the at most $p$ queries to $F$. At this point, $F$ has not been reprogrammed, so it and the masked bit can be simulated without the hidden tuple $(a_1,x_1,\ldots,a_\secp,x_\secp)$. \Cref{cor:hidden-pair-parallel-repetition} bounds the probability that any fixed candidate tuple equals $(a_1,x_1,\ldots,a_\secp,x_\secp)$ by $\beta_{r,p,\varepsilon,N}^{\secp}$. A union bound therefore gives probability at most $p\beta_{r,p,\varepsilon,N}^{\secp}$ that any query address equals the hidden tuple. Applying the same hybrid estimate as above yields
\[
  \left|\Pr[\Hyb(j)=1]-\Pr[\Hyb(j+1)=1]\right|
  \leq2\sqrt{p\beta_{r,p,\varepsilon,N}^{\secp}}.
\]
In $\Hyb(r+1)$, the mask is uniform and independent of the adversary's view apart from the masked bit, so the winning probability is exactly $1/2$. Summing the $r$ adjacent-hop bounds gives the numerical bound in \Cref{eq:security-bound}.

To establish negligibility under the stated hypotheses, write $N:=2^k$
and $\delta_\secp:=(p-1)2^{-\secp/2}=o(1)$. Our parameter choice gives $N\geq4$ and $\sqrt N>(T+1)/4$. Therefore,
\[
\begin{aligned}
  \beta_{r,p,\varepsilon,N}
  &=\frac{1}{N}
    +2r\sqrt{\frac{1+\delta_\secp}{N}}\\
  &\leq\frac{1}{4}
    +\frac{8r\sqrt{1+\delta_\secp}}{T+1}
   =\frac14+o(1).
\end{aligned}
\]
Thus $\beta_{r,p,\varepsilon,N}\leq1/2$ for all sufficiently large $\secp$, and
\[
  2r\sqrt{p\beta_{r,p,\varepsilon,N}^{\secp}}
  \leq 2r\sqrt p\,2^{-\secp/2}
  =\mathrm{negl}(\secp),
\]
since $r$ and $p$ are polynomially bounded. This proves \Cref{eq:security-bound}.
\end{proof}

\section*{Acknowledgements}
Prabhanjan Ananth is supported by the National Science Foundation under grants FET-2329938, CAREER-2341004, and FET-2530160.

\printbibliography

\end{document}